\documentclass[12pt]{article}

\usepackage{amsfonts, amsmath, amssymb, amsthm}
\usepackage{a4}

\DeclareMathOperator{\diag}{diag}
\DeclareMathOperator{\const}{const}

\newtheorem{theorem}{Theorem}
\newtheorem{lemma}{Lemma}

\theoremstyle{definition}

\newtheorem{cnv}{Convention}

\theoremstyle{remark}
\newtheorem{remark}{Remark}
\newtheorem{example}{Example}

\author{I.G. Korepanov}
\title{Multiplicative expression for the coefficient in fermionic 3--3 relation}
\date{March 2015}

\begin{document}

\sloppy

\maketitle

\begin{abstract}
Recently, a family of fermionic relations were discovered corresponding to Pachner move 3--3 and parameterized by complex-valued 2-cocycles, where the weight of a pentachoron (4-simplex) is a Grassmann--Gaussian exponent. Here, the proportionality coefficient between Berezin integrals in the l.h.s.\ and r.h.s.\ of such relations is written in a form multiplicative over simplices.
\end{abstract}

\tableofcontents

\section{Introduction}\label{s:i}

This paper continues the series of papers \cite{K_1301}, \cite{KS2} and~\cite{full-nonlinear}. The reader is referred especially to~\cite{full-nonlinear} for definitions and facts that are only briefly mentioned here. Also, the reader is referred to~\cite{B} for a concise exposition of Grassmann--Berezin calculus of anticommuting variables (or to~\cite{B-super} for a more modern and detailed exposition), and to~\cite{Lickorish} for a pedagogical introduction to Pachner moves.

In paper~\cite{KS2}, a large family was discovered of Grassmann--Gaussian relations corresponding to Pachner move 3--3, with pentachoron (4-simplex) weights depending on a single Grassmann variable attached to each 3-face. In paper~\cite{full-nonlinear}, a full parameterization was given for (a Zariski open set of) such relations, in terms of a 2-cocycle given on both l.h.s.\ and r.h.s.\ of the Pachner move. Many questions still remain, however, to be solved before we arrive at a full-fledged four-dimen\-sional topological quantum field theory on piecewise-linear manifolds.

In the present paper, we solve one of such questions, and show that the answer is remarkably nontrivial. It consists in finding the coefficient called `$\const$' in~\cite[formula~(53)]{full-nonlinear} (as well as~\cite[formula~(6)]{KS2}) in a form that would make possible further construction of a manifold invariant. Namely, the coefficient should be represented as a ratio, $\const=c_r/c_l$ (compare relation~\eqref{33} below), of two expressions belonging to the two sides of the move, and each of these must be \emph{multiplicative} --- have the form of a product over simplices belonging to the corresponding side. This was the case in an earlier paper~\cite{K_1301}, see formula~(1) and Theorem~1 there, also reproduced in~\cite[Section~6]{KS2}, although the 3--3 relations in these papers must be regarded as degenerate from the viewpoint of the present paper. This was also the case in \cite[formula~(38)]{K_1105} and~\cite[formula~(12)]{K_1201}, where different but similar relations were considered.

\subsection{PL manifold invariants and Pachner moves}

In order to construct invariants of piecewise linear (PL) manifolds, it makes sense to construct algebraic relations corresponding to \emph{Pachner moves}, see, for instance, \cite[Section~1]{Lickorish}. Pachner's theorem states that a triangulation of a PL manifold can be transformed into any other triangulation using a finite sequence of these moves~\cite{Pachner}, so there is a hope to pass then from such relations to some quantities characterizing the whole manifold.

In the four-dimensional case, the Pachner moves are 3--3, 2--4 and 1--5. The first of them is usually regarded as `central', and we will be dealing with it in this paper. Here we describe this move and fix notations for the involved vertices and simplices.

Let there be a cluster of three pentachora (4-simplices) 12345, 12346 and~12356 situated around the 2-face~123. Move 3--3 transforms it into the cluster of three other pentachora, 12456, 13456 and~23456, situated around the 2-face~456. The inner 3-faces (tetrahedra) are 1234, 1235 and~1236 in the l.h.s., and 1456, 2456 and~3456 in the r.h.s. The boundary of both sides consists of nine tetrahedra.

Note that we have listed in the previous paragraph exactly \emph{all} simplices in which the l.h.s.\ of move 3--3 differs from its r.h.s. And the boundary of both sides is, of course, the same, it consists of nine tetrahedra.

\subsection{Discrete field theory}

Our relation corresponding to Pachner move 3--3 (also appearing in (practically) this general form in \cite{K_1301,KS2} and~\cite{full-nonlinear}), is
\begin{multline}\label{33}
c_l \iiint \mathcal W_{12345} \mathcal W_{12346} \mathcal W_{12356} \,\mathrm dx_{1234} \,\mathrm dx_{1235} \,\mathrm dx_{1236} \\
 = c_r \iiint \mathcal W_{12456} \mathcal W_{13456} \mathcal W_{23456} \,\mathrm dx_{1456} \,\mathrm dx_{2456} \,\mathrm dx_{3456}.
\end{multline}
Here the integrals are Berezin integrals~\cite{B,B-super} in Grassmann (anticommuting) variables, and $\mathcal W_{ijklm}$ are Grassmann--Gaussian pentachoron weights explained below.

\subsection{The results of this paper, and how they are explained}

The results are explicit formulas for everything in~\eqref{33}: Grassmann weights~$\mathcal W_{ijklm}$ and coefficients $c_l$ and~$c_r$ --- in terms of a 2-cocycle~$\omega$, in accordance with~\cite{full-nonlinear}. As all formulas are algebraic, the author might have presented just these formulas, saying: and now the validity of~\eqref{33} can be checked using computer algebra. The formulas look, however, rather intricate, so we follow another way, focusing on the actual author's reasonings.

\section{Explicit formulas for matrix elements}\label{s:e}

In this Section, as well as in the next Sections \ref{s:d} and~\ref{s:v}, we work within a single pentachoron $u=12345$. The changes to be made for other~$u$ are quite simple and will be explained later.

\begin{cnv}\label{cnv:n}
We denote triangles (2-simplices) by the letter~$s$, tetrahedra (3-simplices) by~$t$, and pentachora (4-simplices) --- by~$u$. As for edges (1-simplices), we tend to use the letter~$b$ for them, while vertices (0-simplices) are $i,j,k,\dots$.
\end{cnv}

\begin{cnv}\label{cnv:o}
We also write the simplices by their vertices, e.g., $s=ijk$ or, as we have written above, $u=12345$. The vertices are thus given by their numbers, and in writing so, we assume by default that the vertices are \emph{ordered}: $i<j<k$, etc. If, however, we need a triangle whose order of vertices in unknown or unessential, we write $s$ as~$\{ijk\}$, as in Lemma~\ref{l:s} below.
\end{cnv}

\subsection{Edge operators}

Our Grassmann--Gaussian pentachoron weight is
\[
\mathcal W_u = \mathcal W_{12345} = \exp\left(-\frac{1}{2} \,\mathsf x^{\mathrm T} F\/ \mathsf x\right),
\]
where
\begin{equation}\label{x}
\mathsf x = \begin{pmatrix} x_{2345} & x_{1345} & x_{1245} & x_{1235} & x_{1234} \end{pmatrix}^{\mathrm T}
\end{equation}
is the column of Grassmann variables corresponding to the 3-faces $t\subset u$, and $F$ --- a $5\times 5$ antisymmetric matrix.

We are going to recall the construction of matrix~$F$ from~\cite{full-nonlinear}. Moreover, we will write down some specific explicit expressions for the entries of~$F$ that do not appear in~\cite{full-nonlinear}. On the other hand, we will skip some details for which the reader is referred to~\cite{full-nonlinear}.

Our starting point is a \emph{2-cocycle}~$\omega$: it takes complex values $\omega_s=\omega_{ijk}$ on triangles~$s=ijk\subset u$ such that
\begin{equation}\label{c}
\omega_{jkl}-\omega_{ikl}+\omega_{ijl}-\omega_{ijk}=0.
\end{equation}
Then, there are \emph{edge operators}~$d_b$ for the ten edges $b=ij\subset u$ that make the bridge between $\omega$ and matrix~$F$. Edge operators have the following properties:
\begin{itemize}\itemsep 0pt
 \item they belong to the 10-dimensional space of operators
  \begin{equation}\label{bpgx}
  d=\sum_{t\subset u} (\beta_t \partial_t+\gamma_t x_t),
  \end{equation}
 \item more specifically, the sum~\eqref{bpgx} for a given~$d_b$ runs \emph{only over such three tetrahedra~$t$ that $t\supset b$},
 \item each of them annihilates the pentachoron weight:
  \[
  d_b \mathcal W_u = 0,
  \]
 \item they are antisymmetric with respect to changing the edge orientation:
  \begin{equation}\label{da}
  d_{ij}=-d_{ji},
  \end{equation}
 \item they obey the following linear relations for each vertex $i\in u$:
  \begin{equation}\label{dv}
  \sum_{\substack{j\in u\\ j\ne i}} d_{ij} = 0,
  \end{equation}
 \item and there is one more linear relation:
  \begin{equation}\label{5u}
  \sum_{b\subset u} \nu_b d_b = 0,
  \end{equation}
  where $\nu$ is any 1-cocycle such that $\omega$ makes its coboundary:
  \[
  \omega=\delta \nu,\quad \text{i.e.,}\quad \omega_{ijk}=\nu_{jk}-\nu_{ik}+\nu_{ij},
  \]
 \item they form a maximal (5-dimensional) \emph{isotropic} subspace in the (10-dimensional) space of all operators of the form~\eqref{bpgx}, where the scalar product is, by definition, the anticommutator:
  \[
  \langle d', d'' \rangle = [d',d'']_+ = d'd''+d''d'.
  \]
\end{itemize}

\subsection{Partial scalar products of edge operators}

Due to the form~\eqref{bpgx}, we have \emph{$t$-components}
\[
d_b|_t=\beta_t \partial_t+\gamma_t x_t
\]
of edge operators, and the (vanishing) scalar product of two edge operators is a sum over tetrahedra:
\[
0 = \langle d_{b_1},d_{b_2} \rangle = \sum_{t\subset u} \langle d_{b_1},d_{b_2} \rangle_t ,
\]
where $\langle d_{b_1},d_{b_2} \rangle_t$ --- we call it the \emph{partial scalar product} of $d_{b_1}$ and~$d_{b_2}$ with respect to tetrahedron~$t$ --- is by definition the same as $\langle d_{b_1}|_t,d_{b_2}|_t \rangle$.

\begin{lemma}\label{l:s}
Choose a tetrahedron $t\subset u$ and a triangle $\{ijk\}\subset t$ (see Convention~\ref{cnv:o} for this notation). Then the partial scalar product $\langle d_{ij},d_{ik} \rangle_t$ remains the same under any permutation of $i,j,k$.
\end{lemma}

\begin{proof}
Let us prove, for instance, that
\begin{equation}\label{1213}
\langle d_{12},d_{13}\rangle_{1234} = \langle d_{21},d_{23}\rangle_{1234}.
\end{equation}
Setting $i=3$ in~\eqref{dv} and taking its $t$-component, we have (keeping in mind also~\eqref{da}):
\begin{equation}\label{d3}
-d_{13}|_{1234}-d_{23}|_{1234}+d_{34}|_{1234}=0.
\end{equation}
We want to take the scalar product of~\eqref{d3} with~$d_{12}$. As $1234$ is the only tetrahedron common for the edges $12$ and~$34$, and all edge operators are orthogonal to each other, we get
\begin{equation}\label{ortho}
\langle d_{12},d_{34}\rangle_{1234}=\langle d_{12},d_{34}\rangle=0.
\end{equation}
So, the mentioned scalar product, together with~\eqref{ortho}, gives~\eqref{1213} at once.
\end{proof}

\begin{lemma}\label{l:t}
For a tetrahedron $t\subset u$, construct the expression
\begin{equation}\label{s-gen}
\omega_s\langle d_{b_1},d_{b_2} \rangle_t.
\end{equation}
Here tetrahedron~$t$ is considered as oriented, $s$ is any of its 2-faces with the induced orientation, and $b_1,b_2\subset s$ are two edges sharing the same initial vertex (thus also oriented). Then, the expression~\eqref{s-gen} does not depend on a specific choice of $s$, $b_1$ and~$b_2$, and thus pertains solely to~$t$.
\end{lemma}

\begin{proof}
Let us prove, for instance, that
\begin{equation}\label{s}
-\omega_{123}\langle d_{12},d_{13}\rangle_{1234}=\omega_{124}\langle d_{12},d_{14}\rangle_{1234}.
\end{equation}
(the minus sign accounts for opposite orientations of $123$ and~$124$). A small exercise shows that the following linear relation is a consequence of~\eqref{5u}:
\[
-\omega_{123}d_{13}|_{1234}-\omega_{124}d_{14}|_{1234}+\omega_{234}d_{34}|_{1234}=0.
\]
Multiplying this scalarly by~$d_{12}$ and using once again orthogonality~\eqref{ortho}, we get~\eqref{s}.
\end{proof}

\begin{lemma}\label{l:u}
Expression~\eqref{s-gen} also remains the same \emph{for all tetrahedra~$t$} forming the boundary of pentachoron~$u$, if these tetrahedra are oriented consistently (as parts of the boundary~$\partial u$).
\end{lemma}

\begin{proof}
It is enough to consider the situation where $s$ is the common 2-face of two tetrahedra $t,t'\subset u$, and show that
\begin{equation}\label{dbb}
\langle d_{b_1},d_{b_2} \rangle_t = -\langle d_{b_1},d_{b_2} \rangle_{t'}.
\end{equation}
Indeed, as the orientation of~$s$ as part of~$\partial t$ is different from its orientation as part of~$\partial t'$, there are two values~$\omega_s$ differing in sign, and \eqref{dbb} will yield at once that \eqref{s-gen} is the same for $t$ and~$t'$.

To prove~\eqref{dbb}, we note that $t$ and~$t'$ are the only tetrahedra containing both $b_1$ and~$b_2$, so
\[
0=\langle d_{b_1},d_{b_2} \rangle = \langle d_{b_1},d_{b_2} \rangle_t + \langle d_{b_1},d_{b_2} \rangle_{t'}.
\]
\end{proof}

\begin{cnv}\label{cnv:en}
We normalize edge operators in such way that quantity~\eqref{s-gen} becomes unity.
\end{cnv}

Here is the matrix of scalar products $\langle d_a,d_b \rangle_{1234}$ calculated according to Convention~\ref{cnv:en}. The rows (resp.\ columns) correspond to edge~$a$ (resp.~$b$) taking values in lexicographic order: 12, 13, 14, 23, 24, 34:
\begin{equation}\label{sc}
\begin{pmatrix}
 \omega^{-1}_{124}{-}\omega^{-1}_{123} & \omega^{-1}_{123} & -\omega^{-1}_{124} & -\omega^{-1}_{123} & \omega^{-1}_{124} & 0 \\[1ex]
 \omega^{-1}_{123} & -\omega^{-1}_{134}{-}\omega^{-1}_{123} & \omega^{-1}_{134} & \omega^{-1}_{123} & 0 & -\omega^{-1}_{134} \\[1ex]
 -\omega^{-1}_{124} & \omega^{-1}_{134} & \omega^{-1}_{124}{-}\omega^{-1}_{134} & 0 & -\omega^{-1}_{124} & \omega^{-1}_{134} \\[1ex]
 -\omega^{-1}_{123} & \omega^{-1}_{123} & 0 & \omega^{-1}_{234}{-}\omega^{-1}_{123} & -\omega^{-1}_{234} & \omega^{-1}_{234} \\[1ex]
 \omega^{-1}_{124} & 0 & -\omega^{-1}_{124} & -\omega^{-1}_{234} & \omega^{-1}_{124}{+}\omega^{-1}_{234} & -\omega^{-1}_{234}\\[1ex]
 0 & -\omega^{-1}_{134} & \omega^{-1}_{134} & \omega^{-1}_{234} & -\omega^{-1}_{234} & \omega^{-1}_{234}{-}\omega^{-1}_{134} 
\end{pmatrix} .
\end{equation}

\begin{remark}
To calculate \emph{diagonal} elements in~\eqref{sc} is an easy exercise using linear relations similar to~\eqref{d3}.
\end{remark}

\begin{remark}\label{r:45}
As for tetrahedron~1235, we must not only replace `4' by `5' in~\eqref{sc}, but also change all signs --- due to its different orientation! Similarly, analogues of matrix~\eqref{sc} for other tetrahedra can be calculated.
\end{remark}

\subsection{Superisotropic operators and matrix~$F$}\label{ss:si}

\emph{Superisotropic operators} are such operators of the form~\eqref{bpgx} that annihilate the weight~$\mathcal W_u$ and whose \emph{each $t$-component} is isotropic, i.e., either $\gamma_t=0$ or $\beta_t=0$. The rows of matrix~$F$ correspond to superisotropic operators in the following sense: every component of the column
\begin{equation}\label{px}
\mathsf p+F\mathsf x,
\end{equation}
where $\mathsf x$ is given by~\eqref{x} and $\mathsf p$, similarly, by
\[
\mathsf p = \begin{pmatrix} \partial_{2345} & \partial_{1345} & \partial_{1245} & \partial_{1235} & \partial_{1234} \end{pmatrix}^{\mathrm T},
\]
is superisotropic.

We recall~\cite[Subsection~4.2]{full-nonlinear} how superisotropic operators \emph{proportional} to entries of~\eqref{px} are constructed in terms of edge operators. They all are linear combinations written as
\begin{equation}\label{1:f}
g = \sum_{1\le i<j\le 5} \alpha_{ij} d_{ij},\qquad \alpha_{ij}\in\mathbb C.
\end{equation}
First, we choose and fix one of two square roots of each~$\omega_s$:
\[
q_s \stackrel{\rm def}{=} \sqrt{\omega_s}.
\]
Second, we define ``initial''~$\alpha_{ij}$ as
\begin{equation}\label{1:alpha}
\alpha_b = \prod_{\substack{s\supset b\\ \mathrm{or}\;s\cap b=\emptyset}} q_s \,.
\end{equation}

\begin{example}
As the 2-faces $s\subset 12345$ containing edge~$12$ are $123$, $124$ and~$125$, and the only 2-face not intersecting with~$12$ is~$345$, such ``initial'' $\alpha_{12}$ is
\begin{equation*}
\alpha_{12} = q_{123}q_{124}q_{125}q_{345}.
\end{equation*}
\end{example}

Finally, the operator proportional to the $i$-th entry in~\eqref{px}, and thus corresponding to the tetrahedron~$t$ \emph{not containing} the vertex~$i$, is obtained by the following change of signs:
\begin{equation}\label{alpha:znaki}
\alpha_b \text{ \ remains the same if \ } b\subset t, \text{ \ else \ }\alpha_b \mapsto -\alpha_b.
\end{equation}

We want to identify the entries in column~\eqref{px} with the operators given by \eqref{1:f}, \eqref{1:alpha} and~\eqref{alpha:znaki}. Such identifications are determined to within a renormalization $x_t\mapsto x'_t=\lambda_t x_t$ of Grassmann variables, implying also $\partial_i\mapsto \partial'_i=(1/\lambda_i) \partial_i$.

\begin{remark}
This renormalization leads to multiplying matrix~$F$ from \emph{both sides} by the diagonal matrix $\diag(\lambda_{2345}^{-1},\,\dots\,,\lambda_{1234}^{-1})$.
\end{remark}

To fix the mentioned arbitrariness, we choose a distinguished edge~$a$ in every tetrahedron~$t$ and assume that the restriction of~$d_a$ onto~$t$ has a unit coefficient before~$\partial_t$:
\[
d_a|_t=\partial_t+\gamma x_t.
\]
As $\langle \partial_t, x_t \rangle = 1$, this implies
\[
\gamma=\frac{1}{2}\langle d_a, d_a\rangle_t.
\]

\begin{cnv}\label{cnv:d}
In this paper, the distinguished edge~$a$ in any tetrahedron~$t$ will always be the lexicographically first one, for example, $a=12$ in tetrahedron~$t=1234$.
\end{cnv}

We now denote~$g^{(t)}$ the superisotropic operator defined according to \eqref{1:f}, \eqref{1:alpha} and~\eqref{alpha:znaki}. If such operator contains a summand~$\gamma x_{t'}$, then $\gamma=\langle g^{(t)}, d_a\rangle_{t'}$, and if it contains $\beta\partial_t$, then $\beta=2\dfrac{\langle g^{(t)}, d_a\rangle_t}{\langle d_a, d_a\rangle_t}$. Hence, the matrix element $F_{tt'} = \gamma / \beta $ (because the coefficient at~$\partial_t$ must be set to unity, according to~\eqref{px}), i.e.,
\begin{equation}\label{Ftt'}
F_{tt'} = \frac{\langle g^{(t)}, d_a\rangle_{t'}\, \langle d_a, d_a\rangle_t}{2\langle g^{(t)}, d_a\rangle_t}.
\end{equation}
The scalar products are calculated according to~\eqref{sc} and Remark~\ref{r:45}.

\begin{example}
Here is a typical matrix element:
\begin{equation}\label{F12}
F_{12}=F_{2345,1345}=
-\frac{(q_{235}^2-q_{234}^2)}{2
 q_{134}q_{135}q_{234}q_{235}}\cdot \frac{f_{12}^{(\mathrm n)}}{f_{12}^{(\mathrm d)}},
\end{equation}
where
\begin{multline}\label{fn}
f_{12}^{(\mathrm n)} = q_{124}q_{134}q_{235}q_{345} -q_{125}q_{135}q_{234}q_{345} +q_{123}q_{135}^2q_{245} \\
-q_{123}q_{134}^2q_{245}-q_{124}q_{135}q_{145}q_{235}+q_{125}q_{134}q_{145}q_{234}
\end{multline}
and
\begin{multline}\label{fd}
f_{12}^{(\mathrm d)} = q_{125}q_{134}q_{235}q_{345} -q_{124}q_{135}q_{234}q_{345} -q_{124}q_{135}q_{235}q_{245} \\
+q_{125}q_{134}q_{234}q_{245}+q_{123}q_{145}q_{235}^2-q_{123}q_{145}q_{234}^2.
\end{multline}
\end{example}

\section{Divisors of matrix elements}\label{s:d}

The central part of the present work consisted in finding a nice description for poles and zeros of matrix elements~$F_{tt'}$ of the typical form~\eqref{F12}. The point is, of course, that the quantities~$\omega_{ijk}=q_{ijk}^2$ make a cocycle, so there are \emph{dependencies}
\begin{equation}\label{qbd}
q_{jkl}^2-q_{ikl}^2+q_{ijl}^2-q_{ijk}^2=0
\end{equation}
for all tetrahedra~$ijkl$.

\subsection{Variables~$a_{ij}$ and their relation to ``initial''~$\alpha_{ij}$}

Recall that we are working within \emph{one pentachoron}~$12345$. It has ten 2-faces, as well as ten edges. This fact, together with the accumulated experience (compare \cite[formula~(50)]{full-nonlinear}), suggests the idea to introduce a \emph{1-chain}~$a_{ik}$ such that $\omega_{ijk}$ is written as a product of its three values, namely:
\begin{equation}\label{ob}
\omega_{ijk}=a_{ij}a_{ik}a_{jk}.
\end{equation}

Given all~$\omega_{ijk}$, the $a_{ij}$ are found from the system of equations which become linear after taking logarithms and are easily solved. Interestingly, the result is, up to an overall factor, our old alphas from formula~\eqref{1:alpha}:
\begin{equation}\label{a-alpha}
a_{ij} = p \cdot \alpha_{ij},
\end{equation}
where
\[
p = \left( \prod_{\substack{\text{over all 2-faces }ijk\\ \text{of pentachoron }12345}} \omega_{ijk} \right)^{-1/6}.
\]

The cocycle relations are now written (instead of~\eqref{qbd}) as
\begin{equation}\label{abd}
a_{kl}a_{jl}a_{jk}-a_{kl}a_{il}a_{ik}+a_{jl}a_{il}a_{ij}-a_{jk}a_{ik}a_{ij}=0.
\end{equation}

\begin{remark}\label{a-anti}
We do not permute the indices of~$a_{ij}$ in this paper, but if needed, the natural idea is to assume that
\[
a_{ij}=-a_{ji}.
\]
\end{remark}

\subsection{Matrix elements in terms of~$a_{ij}$}

Matrix elements~$F_{tt'}$ can now be calculated in terms of~$a_{ij}$. To be exact, here is what we do: set $\alpha_{ij}=a_{ij}/p$ according to~\eqref{a-alpha}; the value of~$p$ is not of great importance because it will soon cancel out. Then apply formula~\eqref{Ftt'} with~$g^{(t)}$ expressed using \eqref{1:f} and~\eqref{alpha:znaki}; the scalar products are, of course, calculated according to \eqref{sc}, Remark~\ref{r:45}, and~\eqref{ob}. The following two examples show what we get.

\begin{example}
\begin{multline}\label{F12a}
F_{12} = 
\frac{a_{25}a_{35}-a_{24}a_{34}}{2a_{13}a_{14}a_{15}a_{34}a_{35}} \\
\cdot \frac{a_{15}a_{34}a_{35}-a_{14}a_{34}a_{35}+a_{14}a_{15}a_{35}-a_{13}a_{15}a_{35}-a_{14}a_{15}a_{34}+a_{13}a_{14}a_{34}}{a_{25}a_{34}a_{35}-a_{24}a_{34}a_{35}-a_{24}a_{25}a_{35}+a_{23}a_{25}a_{35}+a_{24}a_{25}a_{34}-a_{23}a_{24}a_{34}}\,.
\end{multline}
\end{example}

\begin{example}
\begin{multline}\label{F21a}
F_{21} = 
-\frac{a_{15}a_{35}-a_{14}a_{34}}{2a_{23}a_{24}a_{25}a_{34}a_{35}} \\
\cdot \frac{a_{25}a_{34}a_{35}-a_{24}a_{34}a_{35}+a_{24}a_{25}a_{35}-a_{23}a_{25}a_{35}-a_{24}a_{25}a_{34}+a_{23}a_{24}a_{34}}{a_{15}a_{34}a_{35}-a_{14}a_{34}a_{35}-a_{14}a_{15}a_{35}+a_{13}a_{15}a_{35}+a_{14}a_{15}a_{34}-a_{13}a_{14}a_{34}}\,.
\end{multline}
\end{example}

Of course,
\begin{equation}\label{Fskew}
F_{12}=-F_{21},
\end{equation}
even if it is not immediately obvious from \eqref{F12a} and~\eqref{F21a}. We will shed some light on this by studying the poles and zeros of these expressions.

\subsection{The variety of zeros of the main factor in the denominator of a matrix element as function of six variables}

The main factor in the denominator of~\eqref{F12a} is
\begin{equation}\label{k12}
a_{25}a_{34}a_{35}-a_{24}a_{34}a_{35}-a_{24}a_{25}a_{35}+a_{23}a_{25}a_{35}+a_{24}a_{25}a_{34}-a_{23}a_{24}a_{34},
\end{equation}
and its pleasing feature is that is depends on the six variables~$a_{ij}$ belonging to just one tetrahedron~$2345$. There is just one dependence between these~$a_{ij}$:
\begin{equation}\label{2345}
a_{34}a_{35}a_{45}-a_{24}a_{25}a_{45}+a_{23}a_{25}a_{35}-a_{23}a_{24}a_{34}.
\end{equation}

\begin{lemma}\label{l:sing}
The primary decomposition of the affine algebraic variety determined by \eqref{F12a} and~\eqref{2345}, and lying in the affine space of six variables $a_{23},\dots,a_{45}$, consists of the four irreducible components given by the following primary ideals, which are also already prime:
\begin{gather}
(a_{25}+a_{34},\, a_{24}+a_{35}), \label{2345D-} \\
(a_{35},\, a_{24}), \label{2345_0a} \\
(a_{34},\, a_{25}), \label{2345_0b} 
\end{gather}
and
\begin{multline}\label{2345d}
(a_{24}a_{25}a_{34}-a_{24}a_{25}a_{35}-a_{24}a_{34}a_{35}+a_{25}a_{34}a_{35}+a_{24}a_{25}a_{45}-a_{34}a_{35}a_{45},\\
a_{23}a_{25}a_{34}-a_{23}a_{25}a_{35}+a_{23}a_{25}a_{45}-a_{23}a_{34}a_{45}+a_{25}a_{34}a_{45}-a_{34}a_{35}a_{45},\\
a_{23}a_{24}a_{35}-a_{23}a_{25}a_{35}-a_{23}a_{24}a_{45}+a_{24}a_{25}a_{45}+a_{23}a_{35}a_{45}-a_{24}a_{35}a_{45},\\
a_{23}a_{24}a_{34}-a_{23}a_{25}a_{35}+a_{24}a_{25}a_{45}-a_{34}a_{35}a_{45})
\end{multline}
\end{lemma}

\begin{proof}
Direct calculation using Singular computer algebra system.
\end{proof}

\begin{remark}\label{r:sing}
The reader may notice that some more computer calculations of primary decompositions might have been helpful in the process of doing this work. They are, however, more difficult, and the calculation in Lemma~\ref{l:sing} is typical of what the available computer capabilities allowed us to do. And the goal of this work has been achieved!
\end{remark}

While there is no problem understanding the structure of components \eqref{2345D-}--\eqref{2345_0b}, the component~\eqref{2345d} deserves the following lemma.

\begin{lemma}\label{l:2345d}
The affine algebraic variety determined by the ideal~\eqref{2345d}, and lying in the space of six variables $a_{23},\dots,a_{45}$, admits the following trigonometric parameterization:
\begin{equation}\label{tg}
a_{ij} = c\cdot\tan(x_i-x_j).
\end{equation}
It is thus rational, because parameterization~\eqref{tg} becomes rational if re-written in terms of $c$ and tangents of three independent differences of~$x_i$.
\end{lemma}

\begin{proof}
Direct calculation.
\end{proof}

\subsection{Divisor of a matrix element: almost full description, excluding only subvarieties $a_{ij}=0$}\label{ss:d}

We consider the affine algebraic variety~$\mathcal M$ in the space of \emph{ten} variables $a_{12},\dots,a_{45}$, defined by the relations~\eqref{abd} for all tetrahedra $ijkl\subset u=12345$. Then, we consider its Zariski open subspace~$\mathcal M'$ defined as follows:
\begin{equation}\label{M'}
\mathcal M' = (\mathcal M \text{ \ minus all subvarieties where some \ }a_{ij}=0).
\end{equation}

\begin{remark}
The goal of this paper consists in finding the expressions \eqref{cl} and~\eqref{cr} below, for example, by guess. It looks hardly possible to guess these expressions based on nothing, but studying divisors on~$\mathcal M'$ proves to be enough for achieving this goal, so, we content ourself with~$\mathcal M'$. Nevertheless, studying divisors on the whole~$\mathcal M$ might be also of interest, because, for instance, \eqref{2345_0a} and~\eqref{2345_0b} lie exactly in $\mathcal M\setminus\mathcal M'$.
\end{remark}

In~$\mathcal M'$, we introduce the following subvarieties of \emph{codimension~1}, denoted as~$D$ with indices because we think of them as Weil divisors:
\begin{itemize}\itemsep 0pt
 \item $D_u$: this is the subvariety given by the old formulas~\eqref{tg}, but now ten of them: $1\le i\le j\le 5$,
 \item $(D_u)_K$: choose now subset $K\subset \{1,2,3,4,5\}$, and define~$(D_u)_K$ by the same formulas~\eqref{tg} except that we change the signs of those~$a_{ij}$ whose exactly one subscript $i$ or~$j$ is in~$K$. We write also $(D_u)_1$, $(D_u)_{12}$, etc.\ instead of $(D_u)_{\{1\}}$, $(D_u)_{\{1,2\}}$, etc.,
 \item $D_t^-$: for a tetrahedron~$t=ijkl\subset u$, let $b=ij$ be the distinguished edge. Then $D_t^-$ is given by the following equations (compare to~\eqref{2345D-}!):
  \begin{equation}\label{D-}
  D_t^- \colon \quad \begin{cases}a_{ik}=-a_{jl}, \\ a_{jk}=-a_{il}, \end{cases}
  \end{equation}
 \item $D_t^+$, similarly:
  \begin{equation}\label{D+}
  D_t^+ \colon \quad \begin{cases}a_{ik}=a_{jl}, \\ a_{jk}=a_{il}. \end{cases}
  \end{equation}
\end{itemize}

\begin{lemma}\label{l:sigma}
For a tetrahedron $t=ijkl$ and its distinguished edge~$ij$, the sum $D_t^- + D_t^+$ gives, on~$\mathcal M'$, exactly the zero divisor of $\dfrac{\sigma_t}{a_{ij}}=a_{jk}a_{ik}-a_{jl}a_{il}$ (compare with the first factor in the numerator of either \eqref{F12a} or~\eqref{F21a}!), where we denoted
\begin{equation}\label{sigma}
\sigma_t=\omega_{ijk}-\omega_{ijl}.
\end{equation}
\end{lemma}

\begin{proof}
Due to the cocycle relation~\eqref{c}, 
\[
(\sigma_t=0 \text{ \ on \ }\mathcal M') \Leftrightarrow (a_{ik}a_{jk}-a_{il}a_{jl}=0 \text{ \ and \ } a_{ik}a_{il}-a_{jk}a_{jl}=0),
\]
and the r.h.s.\ clearly gives \eqref{D-} or~\eqref{D+}.
\end{proof}

\begin{lemma}\label{l:psi}
For every triangle $s=ijk$, introduce the quantity
\[
\psi_s=a_{jk}-a_{ik}+a_{ij}.
\]
Trigonometric parameterization~\eqref{tg} specifies, on the subset where all $a_{ij}\ne 0$, the variety that can be given by the system of equation of the following form:
\[
\frac{\psi_s}{\omega_s} \text{ \ is the same for all \ } s.
\]
This applies to the case where the indices in~\eqref{tg} take four (like in Lemma~\ref{l:2345d}) as well as five values (or, in fact, any number of them).
\end{lemma}

\begin{proof}
Direct calculation.
\end{proof}

\begin{theorem}\label{th:zp1}
 \begin{enumerate}
  \item\label{i:1} The pole divisor of matrix element~\eqref{F12a}, restricted to~$\mathcal M'$, is~$D_u$.
  \item\label{i:2} The zero divisor of~\eqref{F12a}, restricted to~$\mathcal M'$, is $(D_u)_{12} + D_{2345}^+ + D_{1345}^+$ (the last two are defined in~\eqref{D+}).
 \end{enumerate}
\end{theorem}

\begin{proof}
First, note that the component~\eqref{2345D-} of the divisor of function~\eqref{k12} cancels out with the first factor in the numerator of~\eqref{F12a}, that is,
\begin{equation}\label{n1}
a_{25}a_{35}-a_{24}a_{34},
\end{equation}
and what remains of the zero divisor of~\eqref{n1} after this canceling is~$D_{2345}^+$, according to Lemma~\ref{l:sigma}.

For item~\ref{i:1}, this means that, on the pole variety of~\eqref{F12a}, all the expressions $\psi_s/\omega_s$ are the same for $s\subset 2345$. And analyzing~\eqref{F21a} similarly (and taking into account~\eqref{Fskew}), we arrive at the conclusion that the same are also $\psi_s/\omega_s$ for $s\subset 1345$. It is not hard to deduce now (through a small calculation) that $\psi_s/\omega_s$ are the same \emph{for the whole pentachoron~$12345$, including $s=123$, $124$ and~$125$}. So, item~\ref{i:1} is proved.

For item~\ref{i:2}, we first notice that the main factor in the numerator of~$F_{12}$ (resp.~$F_{21}$) is the same (up to an overall sign) as the main factor in the denominator of~$F_{21}$ (resp.~$F_{12}$) except that the sign is changed of all~$a_{ij}$ with $i=1$ (resp.~$i=2$). For~$F_{12}$, this means that $D_{1345}^+$ appears as a component of zero divisor, in analogy with~\eqref{2345D-}, while the first paragraph of this proof means that $D_{2345}^+$ is also there. The rest, namely $(D_u)_{12}$ appears in full analogy with $D_u$ in the previous paragraph. So, item~\ref{i:2} is also proved.
\end{proof}

\begin{remark}\label{r:zp1}
Theorem~\ref{th:zp1} speaks about a specific matrix element and divisors. It applies, however, to all similar objects, with obvious changes.
\end{remark}

\section{Function $\varphi_{12345}$}\label{s:v}

On our subvariety $\mathcal M'\subset \mathcal M$~\eqref{M'}, we can express all~$a_{ij}$ in terms of~$q_{ijk}$ according to~\eqref{a-alpha}, where the factor~$p$ never vanishes and can be ignored as long as we are considering the zero or pole varieties of expressions \emph{homogeneous} in variables~$a_{ij}$.

\begin{remark}
And all functions of~$a_{ij}$ or~$q_{ijk}$ in this paper \emph{are} homogeneous.
\end{remark}

\begin{remark}
Also, the fact that $p$ is multivalued makes no obstacle on our way.
\end{remark}

\begin{cnv}\label{cnv:aq}
We will denote, taking some liberty, that part of the variety of variables~$q_{ijk}$, $1\le i<j<k\le 5$, where all $q_{ijk}\ne 0$, by the same letter~$\mathcal M'$ as the similar variety in variables~$a_{ij}$. It is implied of course that the~$q_{ijk}$ obey the cocycle relations~\eqref{qbd}. Also, we will use the old notations like $(D_u)_K$ and~$D_t^{\pm}$ for codimension one subvarieties in~$\mathcal M'$ that are like in Subsection~\ref{ss:d} except that we made the substitution~\eqref{a-alpha} in their defining equations.
\end{cnv}

For every 3-face~$t$ of pentachoron~$12345$, we define expression~$f^{(t)}$ as the biggest factor in the denominator of type~\eqref{F12}, namely:
\begin{gather}\label{ft}
f^{(2345)}=q_{125}q_{134}q_{235}q_{345}-q_{124}q_{135}q_{234}q_{345}-q_{124}q_{135}q_{235}q_{245} \nonumber\\ +q_{125}q_{134}q_{234}q_{245} +q_{123}q_{145}q_{235}^2-q_{123}q_{145}q_{234}^2, \label{ft1}\\[.5ex]
f^{(1345)}=q_{124}q_{134}q_{235}q_{345}-q_{125}q_{135}q_{234}q_{345}-q_{123}q_{135}^2q_{245} \nonumber\\ +q_{123}q_{134}^2q_{245} +q_{124}q_{135}q_{145}q_{235}-q_{125}q_{134}q_{145}q_{234}, \\[.5ex]
f^{(1245)}=q_{123}q_{125}^2q_{345}-q_{123}q_{124}^2q_{345}-q_{124}q_{134}q_{235}q_{245} \nonumber\\ +q_{125}q_{135}q_{234}q_{245} -q_{125}q_{134}q_{145}q_{235}+q_{124}q_{135}q_{145}q_{234}, \\[.5ex]
f^{(1235)}=q_{124}q_{125}^2q_{345}-q_{123}^2q_{124}q_{345}-q_{123}q_{134}q_{235}q_{245} \nonumber\\ -q_{125}q_{134}q_{135}q_{245} +q_{125}q_{145}q_{234}q_{235}+q_{123}q_{135}q_{145}q_{234}, \\[.5ex]
f^{(1234)}=q_{124}^2q_{125}q_{345}-q_{123}^2q_{125}q_{345}-q_{123}q_{135}q_{234}q_{245} \nonumber\\ -q_{124}q_{134}q_{135}q_{245} +q_{124}q_{145}q_{234}q_{235}+q_{123}q_{134}q_{145}q_{235} \label{ft5} .
\end{gather}

\begin{remark}
The overall sign of any of expressions \eqref{ft1}--\eqref{ft5} is not now of big importance.
\end{remark}

Also, for every subset $K\subset \{i,j,k,l\}$ consider function~$f_K^{(t)}$ made from \eqref{ft1}--\eqref{ft5} as follows: change the signs at those~$q_{ijk}$ having an odd number of subscripts (one or all three of $i$, $j$ and~$k$) is in~$K$. In the next Lemma~\ref{l:d} we go through the 3-faces of~$12345$ in their natural order, and write down the zero divisors of some interesting functions on~$\mathcal M'$.

\begin{lemma}\label{l:d}
\begin{itemize}\itemsep 0pt
 \item $f^{(2345)}$ has zero divisor $D_u+(D_u)_1+D_{2345}^-$,
 \item $f_1^{(1345)}$ has zero divisor $(D_u)_1+(D_u)_{12}+D_{1345}^+$,
 \item $f_{12}^{(1245)}$ has zero divisor $(D_u)_{12}+(D_u)_{123}+D_{1245}^-$,
 \item $f_{123}^{(1235)}$ has zero divisor $(D_u)_{123}+(D_u)_{1234}+D_{1235}^+$,
 \item $f_{1234}^{(1234)}=f^{(1234)}$ has zero divisor $(D_u)_{1234}+(D_u)_{12345}+D_{1234}^-=(D_u)_5+D_u+D_{1234}^-$.
\end{itemize}
Thus, on~$\mathcal M'$, the function
\begin{equation}\label{fff/ff}
\dfrac{f^{(2345)}f_{12}^{(1245)}f^{(1234)}}{f_1^{(1345)}f_{123}^{(1235)}}
\end{equation}
has the divisor (zeros with sign plus, poles with sign minus)
\begin{equation}\label{MD}
2D_u + D_{2345}^- - D_{1345}^+ + D_{1245}^- - D_{1235}^+ + D_{1234}^-.
\end{equation}
\end{lemma}

\begin{proof}
The formulas for divisors of the first five functions make a simple variations on the theme of Lemma\ref{l:sing}, where, of course, Convention~\ref{cnv:aq} must be also taken into account. Then, \eqref{MD} follows by adding/subtracting relevant divisors.
\end{proof}

Motivated by Lemma~\ref{l:sigma}, we divide the expression~\eqref{fff/ff} by
\[
\sigma_{2345}\sigma_{1245}\sigma_{1234} = (\omega_{234}-\omega_{235})(\omega_{124}-\omega_{125})(\omega_{123}-\omega_{124}).
\]

\begin{theorem}\label{th:vphi}
The divisor of the so obtained expression
\begin{equation}\label{varphi}
\varphi_{12345}=\frac{f^{(2345)}f_{12}^{(1245)}f^{(1234)}} {\sigma_{2345}\sigma_{1245}\sigma_{1234}f_1^{(1345)}f_{123}^{(1235)}},
\end{equation}
considered as a function on~$\mathcal M'$, is
\begin{equation}\label{Dvphi}
2D_u - \sum_{t\subset u} D_t^+ .
\end{equation}
\end{theorem}

\begin{proof}
This follows from~\eqref{MD} and Lemma~\ref{l:sigma}.
\end{proof}

The symmetry of divisor~\eqref{Dvphi} suggests the following theorem.

\begin{theorem}\label{th:v}
Function $\varphi_{12345}$~\eqref{varphi} remains the same, to within a possible sign change, under any permutation of vertices $1,\dots,5$ (the sequences of vertices, both in subscripts and superscripts, are then ordered, according to Convention~\ref{cnv:o}).
\end{theorem}

\begin{proof}
Direct calculation.
\end{proof}

Function~$\varphi_{12345}$ is thus an interesting highly symmetric function of variables~$q_{ijk}$ belonging to the pentachoron~$u=12345$ and obeying the restrictions~\eqref{qbd}.

\section{The poles and zeros of the coefficient in 3--3 relation, and its explicit form}\label{s:c}

We now pass from the single pentachoron~$12345$ to Pachner move 3--3, where six pentachora are involved.

The l.h.s.\ and r.h.s.\ of move 3--3 are triangulated manifolds with boundary. We can orient the pentachora in both these manifolds consistently, and also so that these orientations induce the same orientation on the boundary $\partial(\text{l.h.s.})=\partial(\text{r.h.s.})$. For one such orientation (of two), the signs in the following table show whether this consistent orientation of a pentachoron coincides with the orientation given by the natural order of its vertices: 
\begin{equation}\label{p-ori}
\begin{array}{ccccccc} \multicolumn{3}{c}{\text{left-hand side}} && \multicolumn{3}{c}{\text{right-hand side}} \\ 12345 & 12346 & 12356 && 12456 & 13456 & 23456 \\ + & - & + && + & - & + \end{array}
\end{equation}

We do now all calculations in terms of variables~$q_{ijk}$ and not~$a_{ij}$. This is due to the following important remark.

\begin{remark}
Variables~$a_{ij}$ depend on a pentachoron (i.e., two~$a_{ij}$ for the same edge~$ij$, but calculated within two different pentachora containing this edge, are different), while $q_{ijk}$ do not.
\end{remark}

\subsection{Matrix elements for all six pentachora involved in move 3--3}\label{ss:allp}

In Section~\ref{s:e}, we explained how to calculate matrix~$F$ elements for pentachoron~$12345$. For \emph{any} pentachoron~$ijklm$ (recall that $i<\dots<m$, according to Convention~\ref{cnv:o}), the obvious substitution $1\mapsto i,\dots,5\mapsto m$ must be made. Besides this, the sign of matrix element must be changed for the pentachora marked with minus sign in table~\eqref{p-ori}, as we are going to explain in (the proof of) Lemma~\ref{l:m}, where we study the way how our normalization of edge operators, given by Convention~\ref{cnv:en}, propagates from one pentachoron to another.

\begin{lemma}\label{l:m}
Expression~\eqref{s-gen} can be normalized to unity for a whole oriented triangulated manifold.
\end{lemma}

\begin{proof}
Let tetrahedron~$t$ be the common 3-face of two pentachora, $t=u_1\cap u_2$. Let $a\subset t$ be its edge, and $d_a^{(u_1)}$ and~$d_a^{(u_1)}$ --- the corresponding edge operators in our two pentachora. Then,
\begin{equation}\label{agree-t}
\begin{array}{rc}
 \text{if } & d_a^{(u_1)}|_t=\beta_t\partial_t+\gamma_t x_t, \\
 \text{then   } & d_a^{(u_2)}|_t=\beta_t\partial_t-\gamma_t x_t,
\end{array}
\end{equation}
see~\cite[formulas~(58)]{full-nonlinear}).

We see now that, on passing to a neighboring pentachoron, first, the orientation of~$t$ changes (and this affects the orientation of~$s$ in\eqref{s-gen}!), and second --- partial scalar products of edge operators also change their signs because of~\eqref{agree-t}. Hence, the quantity~\eqref{s-gen} remains the same, as before in Lemmas \ref{l:s}, \ref{l:t} and~\ref{l:u}. This means that it pertains to the whole triangulated manifold, if it is orientable and connected. Hence, we can normalize all edge operators globally so that quantity~\eqref{s-gen} stays always equal to unity.
\end{proof}

And it is clear from~\eqref{Ftt'} that, indeed, changing the sign of partial scalar products implies changing the sign of matrix elements.

\subsection{Components in the l.h.s.\ and r.h.s., their poles and zeros}\label{ss:pz}

Triple integrals in~\eqref{33} are polynomials in Grassmann variables, and their coefficients are proportional. A typical coefficient, namely one at~$x_{1245}$ (the Grassmann variable corresponding to tetrahedron~$1245$), is
\begin{equation}\label{L}
L = F_{1234,1236}F_{1235,1245}-F_{1235,1236}F_{1234,1245}.
\end{equation}
in the l.h.s., and
\begin{equation}\label{R}
R = F_{1456,3456}F_{2456,1245}-F_{2456,3456}F_{1456,1245}.
\end{equation}
in the r.h.s.

\begin{remark}
Two tetrahedra in the subscripts of a matrix element in \eqref{L} or~\eqref{R} clearly determine the relevant pentachoron.
\end{remark}

Our goal is now to guess the form of $c_l$ and~$c_r$. As we can then check the correctness of our guess with a direct computer calculation, informal reasoning will be quite enough for us at this moment.

So, we analyze poles and zeros of $L$, $R$, and other similar Grassmann polynomial coefficients, in order to invent such $c_l$ and~$c_r$ that will compensate these poles and zeros. First, we do so assuming that no one of values~$q_{ijk}$ vanish, that is, within the `global analogue' of set~$\mathcal M'$~\eqref{M'}. The poles of at least one component in the l.h.s.\ are relevant, while the zeros must be \emph{common} for all components; similarly for r.h.s. We see this way that the poles are situated on divisors~$D_u$ (see Subsection~\ref{ss:d}) for \emph{all} pentachora in the relevant side of Pachner move, while the zeros are situated on divisors~$D_t^+$ of all \emph{inner} tetrahedra, again in the relevant side of Pachner move.

\subsection{Fitting the divisors, and the formulas for $c_l$ and~$c_r$}\label{ss:fit-div}

The above analysis of poles and zeros of triple Berezin integrals in~\eqref{33}, when confronted with the divisor~\eqref{Dvphi} of function~$\varphi_{12345}$, suggests that \emph{square roots} of such functions may be the key ingredient of our $c_l$ and~$c_r$. So, we introduce, in analogy with~$\varphi_{12345}$, quantities~$\varphi_u$ for each pentachoron~$u$ (simply making relevant subscript substitutions).

Now we look at what may happen where some $q_{ijk}$ do vanish. Motivated by the products of~$q_{ijk}$ factored out in the denominators of expressions like~\eqref{F12}, we introduce the quantities
\begin{equation}\label{varrho}
\varrho_t=q_{ijk}q_{ijl}
\end{equation}
for tetrahedra $t=ijkl$. These~$\varrho_t$ are expected to compensate the poles appearing where the mentioned denominators vanish.

\begin{remark}
Note that 2-faces $ijk$ and~$ijl$ in~\eqref{varrho} both contain the distinguished edge~$ij\subset t$, see Subsection~\ref{ss:si}.
\end{remark}

What remains is a bit more guessing, trying, scrutinizing formulas --- and we arrive at the following theorem.

\begin{theorem}\label{th:e}
For the 3--3 relation~\eqref{33} to hold, it is enough to set its left-hand-side coefficient to
\begin{equation}\label{cl}
c_l = \varrho_{1234}\, \varrho_{1235}\, \varrho_{1236}\, \sqrt{\varphi_{12345}}\, \sqrt{\varphi_{12346}}\, \sqrt{\varphi_{12356}}\,/\,q_{123}\, ,
\end{equation}
and its right-hand-side coefficient to
\begin{equation}\label{cr}
c_r = \varrho_{1456}\, \varrho_{2456}\, \varrho_{3456}\, \sqrt{\varphi_{12456}}\, \sqrt{\varphi_{13456}}\, \sqrt{\varphi_{23456}}\,/\,q_{456}\, .
\end{equation}
\end{theorem}

\begin{proof}
Direct calculation.
\end{proof}

A small miracle in formulas \eqref{cl} and~\eqref{cr} is the denominators $q_{123}$ and~$q_{456}$, appearing because exactly such a value factors out in a non-obvious way in the numerator of $L$ or~$R$ during the reduction to common denominator.

\subsection*{Acknowledgements}

I thank Evgeniy Martyushev for his interest in this work. He made also some interesting calculations; although they are not used directly in this paper, they showed me some beautiful things and added thus to my inspiration.

I am also grateful to creators and maintainers of Maxima and Singular computer algebra systems for their great work.

\end{document}